\documentclass[11pt,a4paper]{article}

\usepackage[truedimen,margin=34mm]{geometry} 

\usepackage{mathrsfs}
\usepackage{amssymb}
\usepackage{amsmath}
\usepackage{ascmac}
\usepackage{amsthm}
\usepackage[pdftex]{graphicx}
\usepackage[pdftex]{color}
\usepackage{natbib}
\usepackage{setspace}
\usepackage{color}

\usepackage{titlesec}
\titleformat*{\section}{\large\bfseries}
\titleformat*{\subsection}{\it}

\newtheorem{thm}{Theorem}
\newtheorem{lem}{Lemma}

\def\ep{{\varepsilon}}

\def\th{{\theta}}

\def\phih{{\widehat \phi}}

\def\tht{{\widetilde \th}}

\def\tr{{\rm tr\,}}

\def\T{{ \mathrm{\scriptscriptstyle T} }}

\def\ep{{\varepsilon}}
\def\th{{\theta}}

\def\beh{{\hat \beta}}
\def\thh{{\hat \th}}
\def\phih{{\hat \phi}}
\def\Ah{\hat{A}}
\def\Mh{\hat{M}}
\def\Var{\text{var}}

\def\tht{\tilde{\th}}

\title{{\bf Robust Empirical Bayes Small Area Estimation with Density Power Divergence}}
\date{}
\author{}

\begin{document}

\maketitle
\doublespacing

\vspace{-1.7cm}
\begin{center}
{\large Shonosuke Sugasawa}

\medskip
Center for Spatial Information Science, The University of Tokyo
\end{center}

\medskip
\begin{center}
{\bf \large Abstract}
\end{center}

\vspace{-0cm}
A two-stage normal hierarchical model called the Fay--Herriot model and the empirical Bayes estimator are widely used to provide indirect and model-based estimates of means in small areas. However, the performance of the empirical Bayes estimator might be poor when the assumed normal distribution is misspecified. In this article, we propose a simple modification by using density power divergence and suggest a new robust empirical Bayes small area estimator. The mean squared error and estimated mean squared error of the proposed estimator are derived based on the asymptotic properties of the robust estimator of the model parameters. We investigate the numerical performance of the proposed method through simulations and an application to survey data.


\newpage
\section{Introduction}
Direct survey estimators based only on area-specific sample data are known to yield unacceptably large standard errors if the area-specific sample sizes are small.
Empirical Bayes methods are widely used to improve direct survey estimators by shrinking toward some synthetic estimator and borrowing strength.
For comprehensive overviews of small area estimation, see \cite{Pfe2013} and \cite{Rao2015}.

A basic area-level model is a two-stage normal hierarchical model known as the Fay--Herriot model \citep{FH1979}, described as 
\begin{equation}\label{FH}
y_i|\th_i\sim N(\th_i,D_i), \hspace{0.5cm} \th_i\sim N(x_i^\T \beta,A) \hspace{0.5cm} (i=1,\ldots,m),
\end{equation}
where $y_i$ is the direct estimator of the small area mean $\th_i$, $D_i$ is the sampling variance, assumed to be known, $x_i$ and $\beta$ are vectors of the covariates and regression coefficients, respectively, and $A$ is an unknown variance. 
Let $\phi=(\beta^\T ,A)^\T $ be the unknown parameter vector in (\ref{FH}).
Since $y_i\sim N(x_i^\T \beta,A+D_i)$ under (\ref{FH}), $\phi$ can be estimated by maximizing the log-marginal likelihood
\begin{equation}\label{LL}
\log f(y; \phi)=-\frac{m}2\log(2\pi)-\frac12\sum_{i=1}^m\log(A+D_i)-\frac12\sum_{i=1}^m\frac{(y_i-x_i^\T \beta)^2}{A+D_i},
\end{equation}
with $y=(y_1,\ldots,y_m)^\T $.
The Bayes predictor of $\th_i$ under squared error loss is 
\begin{equation}\label{BE}
\tht_i(y_i;\phi)=y_i-\frac{D_i}{A+D_i}(y_i-x_i^\T \beta),
\end{equation}
and the empirical Bayes estimator of $\th_i$ is $\thh_i=\tht(y_i;\phih_{\alpha})$.

The empirical Bayes estimator is useful when $y_i$ can be well-explained by the auxiliary information $x_i$.
However, $x_i$ is not necessarily good auxiliary information for $y_i$ in all the areas, that is, $y_i$ could be very far from $x_i^t\beta$ in some areas, which we call outlying observations.
For such observations, the corresponding $\theta_i$ could be generated from a distribution different from the assumed one (\ref{FH}), that is, the assumed distribution of $\theta_i$ could be misspecified.
In this paper, we consider a situation where there exists such outlying observations and focus on the following two undesirable properties in this situation:

\begin{itemize}
\item[1.]
The Bayes predictor (\ref{BE}) might over-shrink outlying $y_i$ toward $x_i^\T \beta$.

\item[2.]
The estimator $\phih$ based on (\ref{LL}) would be highly influenced by outlying observations.
\end{itemize}

These problems have been addressed in studies such as \cite{FH1979} and \cite{Ghosh2008}, but we extend the body of knowledge on this topic by using density power divergence \citep{Basu1998}.

Our insight is based on an alternative expression for the Bayes predictor (\ref{BE}) using (\ref{LL}).
From Tweedie's formula \citep{Efron2011}, the Bayes predictor (\ref{BE}) can be written as
\begin{equation}\label{BE2}
\tht_i(y_i;\phi)=y_i+D_i\frac{\partial}{\partial y_i}\log f(y_i;\phi).
\end{equation}
The above expression holds as long as $y_i|\theta_i\sim N(\theta_i,D_i)$, that is, only the form of marginal likelihood $f(y_i;\phi)$ should be changed when the distribution of $\theta_i$ is not normal as in (\ref{FH}).
From (\ref{BE2}), one can see that the classical empirical Bayes estimator $\thh_i$ can be determined by the maximization and derivative of (\ref{LL}). 
We therefore suggest replacing the log-marginal likelihood with density power divergence, which includes Kullback--Leibler divergence as a special case.
Density power divergence under (\ref{FH}) has a closed form, and a new robust Bayes predictor has a simple form.
We also consider robust estimators of the model parameters and provide their asymptotic properties.
Moreover, we construct an estimator of the mean squared error of the robust empirical Bayes estimator based on the parametric bootstrap and provide its asymptotic validity.

Generalized likelihood including density power divergence has been used in Bayesian inference \citep{AG2013,GB2016,HV2014,Jewson2018,NH2019}, who address the misspecification of the assumed distribution of observations, but we deal with the misspecification of the assumed distribution of unobserved areal mean $\th_i$ in (\ref{FH}), which can be regarded as the prior distribution of $\th_i$.
Moreover, we consider frequentist inference for the model parameters in (\ref{FH}).

\cite{Ghosh2008} proposed a robust Bayes predictor in the Fay--Herriot model (\ref{FH}), using the influence function for $\beta$ to tackle Property 1, but did not address Property 2.
\cite{Sinha2009} proposed using \cite{Huber1973} $\psi$-function to derive a Bayes predictor and parameter estimators in general linear mixed models, which tackled Properties 1 and 2, but as demonstrated in the next section, the resulting Bayes predictor has limitations when aiming to compensate for Property 1.

\section{Density Power Divergence and Bayes Predictor}\label{sec:DPD}

\subsection{Density power divergence}
Although the maximum likelihood estimator minimizes empirical estimates of Kullback--Leibler distance, it is sensitive to distributional assumptions.
To overcome this, \cite{Basu1998} introduced an estimation method based on density power divergence for independently and identically distributed data.
As the observations $y_i$ following the Fay--Herriot model (\ref{FH}) are independent but not identically distributed, we consider the following function instead of the log-likelihood function (\ref{LL}) \citep{GB2013}, 
\begin{equation}\label{DPD}
L_{\alpha}(y; \phi)=\frac{1}{\alpha }\sum_{i=1}^m f_i(y_i;\phi)^{\alpha}-\frac{1}{1+\alpha}\sum_{i=1}^m  \int f_i(t;\phi)^{1+\alpha}dt, \ \ \ \  \alpha\in (0,1)
\end{equation} 
where $f_i(y_i;\phi)$ is the density of $y_i\sim N(x_i^\T \beta,A+D_i)$. 
Here, $\alpha$ is a tuning constant related to robustness.
Note that 
$$
\lim_{\alpha\to 0}\bigg\{L_{\alpha}(y; \phi)-m\bigg(\frac{1}{\alpha}-1\bigg)\bigg\}=\log f(y;\phi),
$$
so apart from an irrelevant constant (\ref{DPD}) is similar to the log-likelihood function when $\alpha\approx 0$.

Under model (\ref{FH}), the $y_i$s are independent and $y_i\sim N(x_i^\T \beta,A+D_i)$, so (\ref{DPD}) can be expressed as 
\begin{equation}\label{NDPD}
\begin{split}
L_{\alpha}(y; \phi)
&=\sum_{i=1}^m\left\{\frac{s_i(y_i;\phi)}{\alpha}-\frac{V_i^{\alpha}}{(1+\alpha)^{3/2}}\right\},
\end{split}
\end{equation}
where $V_i=\{2\pi(A+D_i)\}^{-1/2}$ and 
\begin{equation*}
s_i(y_i;\phi)=V_i^{\alpha}\exp\left\{-\frac{\alpha(y_i-x_i^\T \beta)^2}{2(A+D_i)}\right\}.
\end{equation*}
We propose using function (\ref{NDPD}) instead of the log-marginal likelihood $\log f(y;\phi)$.

\subsection{Robust Bayes predictor}
We define the robust Bayes predictor $\tht_i^{\rm R}$ of $\th_i$ by replacing $\log f(y;\phi)$ in (\ref{BE2}) with $L_{\alpha}(y; \phi)$.
Since
$$
\frac{\partial}{\partial y_i}L_{\alpha}(y; \phi)=\frac{\partial}{\partial y_i}\frac{s_i(y; \phi)}{\alpha}=\frac{1}{A+D_i}(y_i-x_i^\T \beta)s_i(y_i;\phi),
$$
the robust Bayes predictor is 
\begin{equation}\label{RBE}
\tht_i^{\rm R}=y_i-\frac{D_i}{A+D_i}(y_i-x_i^\T \beta)s_i(y_i;\phi).
\end{equation}
The shrinkage factor in (\ref{RBE}) is $s_i(y_i;\phi)D_i/(A+D_i)$, which depends on $y_i$, whereas the shrinkage factor in the classical Bayes predictor (\ref{BE}) is $D_i/(A+D_i)$, which does not depend on $y_i$.
Further, $\tht_i^{R}$ reduces to $\tht_i$ when $\alpha=0$ since $s_i(y_i;\phi)=1$ under $\alpha=0$.
Moreover, as $D_i\to 0$, the robust Bayes predictor $\tht_i^{\rm R}$ reduces to the direct estimator $y_i$ as the classical $\tht_i$ does.

\subsection{Comparison with related robust Bayes predictors}\label{sec:other}
For related robust Bayes predictors under model (\ref{FH}), \cite{Ghosh2008} proposed the predictor
\begin{equation}\label{GEB}
\tht_i^{\rm G}=y_i-\frac{D_iv_i(A)^{1/2}}{A+D_i}\psi_K\left\{\frac{y_i-x_i^\T \beh(A)}{v_i(A)^{1/2}}\right\},
\end{equation}
where 
$$
\beh(A)=\left(\sum_{i=1}^m\frac{x_ix_i^\T }{A+D_i}\right)^{-1}\left(\sum_{i=1}^m\frac{x_iy_i}{A+D_i}\right), 
\ \ \ \ 
v_i(A)=A+D_i-x_i^\T \left(\sum_{i=1}^m\frac{x_ix_i^\T }{A+D_i}\right)^{-1}x_i,
$$ 
and $\psi_K(t)=u\min(1,K/|u|)$ is Huber's $\psi$-function with a tuning constant $K>0$ that has a similar role to $\alpha$.
Similarly, \cite{Sinha2009} used Huber's $\psi$-function to modify an equation for $\th_i$ and suggested a robust predictor $\tht_i^{\rm SR}$ as a solution to the equation
\begin{equation}\label{SR}
D_i^{-1/2}\psi_K\Big\{D_i^{-1/2}(y_i-\tht_i^{\rm SR})\Big\}-A^{-1/2}\psi_K\Big\{A^{-1/2}(\tht_i^{\rm SR}-x_i^\T \beta)\Big\}=0.
\end{equation}

Consider observation $y_i$ which is very different from the grand (prior) mean $x_i^\T\beta$.
For such observation, the auxiliary information $x_i$ would not be useful to improve the direct estimator $y_i$ through the model (\ref{FH}), so that it would be better to keep $y_i$ unshrunk.
To see such shrinkage property, it would be useful to check the behavior of a Bayes predictor $\eta_i$ under large $|y_i-x_i^t\beta|$.
Specifically, we consider $|\eta_i-y_i|$ as $|y_i-x_i^t\beta|\to\infty$ with fixed values of the parameters and $D_i$. 
Ideally, $|\eta_i-y_i|\to 0$, which means that the Bayes predictor $\eta_i$ does not shrink $y_i$ under large $|y_i-x_i^t\beta|$.
This property was addressed in the context of small area estimation \citep{Datta1995} as well as signal estimation \citep{Cal2010}.
For the classical Bayes predictor $\tht_i$ in (\ref{BE}), $|\tht_i-y_i|\to\infty$ as $|y_i-x_i^t\beta|\to\infty$, meaning that over-shrinkage occurs.
For $\tht_i^{G}$, $|\tht_i^{\rm G}-y_i|\to KD_iv_i(A)^{1/2}/(A+D_i)$.
Moreover, if $|y_i-\tht_i^{\rm SR}|\to 0$ as $|y_i-x_i^t\beta|\to\infty$, the left-hand side of (\ref{SR}) reduces to $-A^{-1/2}K$, so $|y_i-\tht_i^{\rm SR}|\nrightarrow 0$.
On the contrary, for the proposed robust Bayes predictor $\tht_i^{\rm R}$ in (\ref{RBE}), $|\tht_i^{\rm R}-y_i|\to 0$ as $|y_i-x_i^t\beta|\to\infty$ holds when $\alpha>0$, since $(y_i-x_i^\T \beta)s_i(y_i;\phi)\to 0$ as $|y_i-x_i^t\beta|\to \infty$.

When there are no random effects, that is, $A=0$, the conventional Bayes predictor (\ref{BE}) reduces to $x_i^\T \beta$.
However, the robust Bayes predictors, $\tht_i^{\rm R}$, $\tht_i^{\rm G}$, and $\tht_i^{\rm SR}$, do not have the property, which might be a drawback as a compensation for robustness.

\section{Robust Empirical Bayes Estimator and Mean Squared Error}\label{sec:REB}

\subsection{Robust parameter estimation}
We define the robust estimator $\phih_{\alpha}$ of $\phi$ as $\phih_{\alpha}=\text{argmax} \ L_{\alpha}(y;\phi)$, where $L_{\alpha}(y;\phi)$ is given in (\ref{NDPD}). 
Then, the robust estimator $\phih_{\alpha}$ satisfies
\begin{equation}\label{EE}
\begin{split}
\frac{\partial L_{\alpha}}{\partial \beta}=
&\sum_{i=1}^m\frac{x_is_i(y_i;\phi)(y_i-x_i^\T \beta)}{A+D_i}=0,\\
2\frac{\partial L_{\alpha}}{\partial A}=
&\sum_{i=1}^m\left\{\frac{(y_i-x_i^\T \beta)^2s_i(y_i;\phi)}{(A+D_i)^2}-\frac{s_i(y_i;\phi)}{A+D_i}+\frac{\alpha V_i^{\alpha}}{(\alpha+1)^{3/2}(A+D_i)}\right\}=0.
\end{split}
\end{equation}
We adopt a Newton--Raphson algorithm for solving these estimating equations, where derivatives are given in the proof of Theorem \ref{thm:asymp} in the Supplementary Material.
A reasonable starting point would be the maximum likelihood estimates.
By substituting the robust estimator $\phih_{\alpha}$ into the robust Bayes predictor (\ref{RBE}), we obtain the robust empirical Bayes estimator $\thh_i^{\rm R}=\tht_i^{\rm R}(y_i,\phih_{\alpha})$.

\subsection{Selection of tuning parameter}\label{sec:selection}
The parameter $\alpha$ is related to robustness but is not easy to interpret.
Following \cite{Ghosh2008}, we consider selection of $\alpha$ based on the mean squared error of the robust Bayes predictor (\ref{RBE}), which enables us to specify $\alpha$ in an interpretable way.
The mean squared error formula is given in the following theorem.

\begin{thm}\label{thm:BEMSE}
Under model (\ref{FH}), $E\{(\tht_i^{\rm R}-\th_i)^2\}=g_{1i}(A)+g_{2i}(A)$, where 
$$
g_{1i}(A)=\frac{AD_i}{A+D_i},    \hspace{0.5cm}
g_{2i}(A)=\frac{D_i^2}{A+D_i}\left\{\frac{V_i^{2\alpha}}{(2\alpha+1)^{3/2}}-\frac{2V_i^{\alpha}}{(\alpha+1)^{3/2}}+1\right\}
$$
and $g_{2i}(A)$ is increasing in $\alpha$.
\end{thm}

\medskip
The mean squared error of the classical Bayes predictor (\ref{BE}) corresponds to $g_{1i}(A)$, so the excess mean squared error of $\tht_i^{\rm R}$ over $\tht_i$ is $g_{2i}(A)$, which approaches $0$ when $\alpha=0$.
Therefore, there is a trade-off between the robustness of $\tht_i^{\rm R}$ and the mean squared error evaluated under model (\ref{FH}).
We define $\text{Ex}(\alpha)=100\times\sum_{i=1}^mg_{2i}(\Ah_{\alpha})/\sum_{i=1}^mg_{1i}(\Ah_{\alpha})$ as the percentage relative excess mean squared errors, where $\Ah_{\alpha}$ is the robust estimate of $A$ from (\ref{EE}) under given $\alpha$.
We propose selecting $\alpha$ such that $\text{Ex}(\alpha)$ does not exceed a user-specified percentage $c\%$; in other words, we compute $\alpha^{\ast}$ to satisfy $\text{Ex}(\alpha^{\ast})=c$.
This has a unique solution since $\text{Ex}(\alpha)$ increases in $\alpha$ from Theorem \ref{thm:BEMSE}.
We adopt the bisectional method \citep[\S 2]{Burden2010} to compute $\alpha^{\ast}$.
In practice, we first compute $\alpha^{\ast}$ for a specified value of $c$, and all the estimation procedures are conducted with $\alpha=\alpha^{\ast}$.
For theoretical simplicity, we assume that $\alpha$ is known in \S \ref{sec:asymp}, \ref{sec:MSE1} and \ref{sec:MSE2}, but the selection procedure is used in all the numerical examples given in \S \ref{sec:num} to investigate its possible effect.

\subsection{Asymptotic properties of the robust estimators}\label{sec:asymp}
We consider the asymptotic properties of the robust estimator under model (\ref{FH}).
To this end, we assume the regularity conditions:

\begin{itemize}
\item[1.]
$0<D_{\ast}\leq \min_{1\leq i\leq m}D_i\leq \max_{1\leq i\leq m}D_i\leq D^{\ast}<\infty$, where $D_{\ast}$ and $D^{\ast}$ do not depend on $m$;
\item[2.]
$\max_{1\leq i\leq m}x_i^\T (X^\T X)^{-1}x_i=O(m^{-1})$, where $X=(x_1,\ldots,x_m)^\T $;

\item[3.]
$X^\T X/m$ converges to a positive definite matrix as $m\to\infty$.
\end{itemize}

Similar conditions are used by \cite{Prasad1990}.
Since the derivatives in (\ref{EE}) have zero expectations under model (\ref{FH}), we obtain the following result.

\begin{thm}\label{thm:asymp}
Under Conditions 1--3, $\beh_{\alpha}$ and $\Ah_{\alpha}$ are asymptotically independent and distributed as $N(\beta,m^{-1} J_\beta^{-1}K_\beta J_\beta^{-1})$ and $N(A,K_A/mJ_A^2)$, respectively, where 
\begin{equation*}\label{AV}
\begin{split}
& J_\beta=\frac{1}{m(\alpha+1)^{3/2}}\sum_{i=1}^m\frac{V_i^{\alpha}x_ix_i^\T }{A+D_i}, \ \ \ \ \ 
J_A=\frac{1}{2m}\sum_{i=1}^m\frac{V_i^{\alpha}(\alpha^2+2)}{(A+D_i)^2(\alpha+1)^{5/2}},\\
&K_\beta=\frac{1}{m(2\alpha+1)^{3/2}}\sum_{i=1}^m\frac{V_i^{2\alpha}x_ix_i^\T }{A+D_i}, \ \ \ 
K_A=\frac{1}{m}\sum_{i=1}^m\frac{V_i^{2\alpha}}{(A+D_i)^2}\left\{\frac{2(2\alpha^2+1)}{(2\alpha+1)^{5/2}}-\frac{\alpha^2}{(\alpha+1)^3}\right\}.
\end{split}
\end{equation*}
\end{thm}

\vspace{0.2cm}
When $\alpha=0$, the asymptotic covariance matrix of $\beh_{\alpha}$ and the asymptotic variance of $\Ah_{\alpha}$ reduce to the asymptotic covariance matrix and variance of the maximum likelihood estimator of $\beta$ and $A$ given by \cite{Datta2000}, because (\ref{NDPD}) reduces to the log-likelihood (\ref{LL}).

\subsection{Mean squared error of the robust empirical Bayes estimator}\label{sec:MSE1}
To evaluate the risk of the estimator $\thh_i^{\rm R}$, we consider the mean squared error, $M_i=E\{(\thh_i^{\rm R}-\th_i)^2\}$, where the expectation is taken with respect to the joint distribution of the $\th_i$s and $y_i$s following the assumed model (\ref{FH}). 
The mean squared error can be regarded as the integrated Bayes risk, and is a standard measure of risk in small area estimation \citep{Rao2015}.

Since $\thh_i^{\rm R}$ depends on the estimator $\phih_{\alpha}$, the mean squared error $M_i$ takes account of the additional variability due to $\phih_{\alpha}$.
Therefore, it is difficult to evaluate $M_i$ analytically, and a second-order approximation of $M_i$ has been used. 
Following this convention, we provide an approximation for $M_i$ in the following theorem.

\begin{thm}\label{thm:mse}
Under Conditions 1--3, 
\begin{equation}\label{mse}
M_i=g_{1i}(A)+g_{2i}(A)+\frac{g_{3i}(A)}{m}+\frac{g_{4i}(A)}{m}+\frac{2g_{5i}(A)}{m}+o(m^{-1}),
\end{equation} 
where $g_{1i}(A)$ and $g_{2i}(A)$ are given in Theorem \ref{thm:BEMSE}, and 
\begin{align*}
g_{3i}(A)&=\frac{D_i^2V_i^{2\alpha}}{B_i^2(2\alpha+1)^{3/2}}x_i^\T  J_{\beta}^{-1}K_{\beta} J_{\beta}^{-1}x_i, \ \ \ \ \ 
g_{4i}(A)=\frac{D_i^2V_i^{2\alpha}K_A}{B_i^3(2\alpha+1)^{7/2}J_A^2}\Big(\alpha^4-\frac12\alpha^2+1\Big),\\
g_{5i}(A)&
=\frac{\alpha D_i^2x_i^\T  J_{\beta}^{-1}K_{\beta} J_{\beta}^{-1}x_i}{2B_i^4}(3B_iC_{11}-\alpha C_{21})
+\frac{D_i^2K_A}{24B_i^6J_A^2}\Big\{3\alpha B_i^2C_{21}+(\alpha-2)(3\alpha+8)C_{11}\Big\}\\
&\ \ \ \ +\frac{D_i^2x_i^\T  J_{\beta}^{-1}x_i}{B_i^4}(B_iC_{12}-\alpha C_{22})
+\frac{D_i^2}{2B_i^6J_A}\Big\{\alpha C_{32}-2B_iC_{22}+(2-\alpha)B_i^2C_{12}\Big\}\\
&\ \ \ \ +\frac{D_i^2}{2B_i^4}\Big\{b_A-\frac{\alpha V_i^{\alpha}}{(\alpha+1)^{3/2}B_iJ_A}\Big\}\Big\{(2-\alpha)B_iC_{11}-\alpha C_{21}\Big\}.
\end{align*}
Here, $B_i=A+D_i$, $b_A=\lim_{m\to\infty}mE(\Ah_{\alpha}-A)$ is the first-order bias of $\Ah_{\alpha}$ and
$$
C_{jk}=(2j-1)!!B_i^j\left\{V_i^{k\alpha}(k\alpha+1)^{-j-1/2}-V_i^{k\alpha+\alpha}(k\alpha+\alpha+1)^{-j-1/2}\right\},
$$
where $(2j-1)!!=(2j-1)(2j-3)\cdots(1)$.
\end{thm}

\vspace{0.3cm}
The derivation is given in the Supplementary Material.
It should be noted that the approximation formula (\ref{mse}) is based on known $\alpha$, so it might be different under estimated (selected) $\alpha$.  
The approximation formula (\ref{mse}) reduces to the mean squared error of the classical empirical Bayes estimator given by \cite{Datta2000} and \cite{Datta2005} since $C_{jk}=0$ and $g_{5i}(A)=0$ under $\alpha=0$.

\subsection{Estimation of the mean squared error}\label{sec:MSE2}
Because the approximation of the mean squared error given in Theorem \ref{thm:mse} depends on the unknown parameter $A$, it cannot be used in practice. 
We use a second-order unbiased estimator of the mean squared error.
An estimator $\hat{T}$ is called second-order unbiased if $E(\hat{T})=T+o(m^{-1})$. 
As shown in Theorem \ref{thm:mse}, $g_{3i}(A), g_{4i}(A)$, and $g_{5i}(A)$ are smooth functions, so $g_{3i}(\Ah_{\alpha}), g_{4i}(\Ah_{\alpha})$, and $g_{5i}(\Ah_{\alpha})$ are second-order unbiased.
On the contrary, $g_{1i}(\Ah_{\alpha})$ and $g_{2i}(\Ah_{\alpha})$ may have considerable bias, since $g_{1i}(A)$ and $g_{2i}(A)$ are $O(1)$.
Since the derivation of these biases and bias-corrected estimators of these terms require tedious algebra, we use the parametric bootstrap method in a similar way to \cite{Butar2003}.
We define the bootstrap estimator 
\begin{equation}\label{bMSE}
\Mh_i=2g_{12i}(\Ah_{\alpha})-\frac1B\sum_{b=1}^Bg_{12i}(\Ah_{\alpha}^{(b)})+\frac{g_{3i}(\Ah_{\alpha})}{m}+\frac{g_{4i}(\Ah_{\alpha})}{m}+\frac{2g_{5i}(\Ah_{\alpha})}{m},
\end{equation}
where $g_{12i}(A)=g_{1i}(A)+g_{i2}(A)$ and $\Ah_{\alpha}^{(b)}$ is the bootstrap estimator based on the parametric bootstrap samples $y_1^{(b)},\ldots,y_m^{(b)}$ generated from 
$$
y_i^{(b)}=x_i^\T \beh_{\alpha}+v_i^{(b)}+\ep_i^{(b)}, \ \ \ \ v_i^{(b)}\sim N(0,\Ah_{\alpha}), \ \ \ \ep_i^{(b)}\sim N(0,D_i).
$$
Following \cite{CH2015}, we obtain the following theorem.

\begin{thm}\label{thm:bMSE}
Assume Conditions 1--3 and let $\Mh_i^{\dagger}$ be the ideal version of $\Mh_i$ obtained by taking $B=\infty$.
Define $\Mh_i=\Mh_i^{\dagger}+U_i$, where $U_i$ denotes an error term arising from doing only a finite number of bootstrap replications.
Then, $E(\Mh_i^{\dagger})-M_i=o(m^{-1})$ and $U_i=O_p\{(mB)^{-1/2}\}$.
\end{thm}

From Theorem \ref{thm:bMSE}, the ideal version of the estimator (\ref{bMSE}), $\Mh_i^{\dagger}$, is second-order unbiased.
Moreover, Theorem \ref{thm:bMSE} implies that if $B=O(m^{1+\delta})$ for some small $\delta>0$, the error from the finite numbers of bootstrap iterations is $o_p(m^{-1})$, and thus the estimator $\Mh_i$ would perform similarly to the ideal estimator $\Mh_i^{\dagger}$.
However, the theoretical result is based on known $\alpha$, so that the bootstrap estimator (\ref{bMSE}) is not necessarily justified under estimated $\alpha$.

Although we adopt an additive form of bias correction in (\ref{bMSE}) following \cite{Butar2003}, other forms of bias correction have been proposed, such as those of \cite{Hall2006}.
Concerning $g_{5i}(\Ah_{\alpha})$, we must compute the estimate of $b_A$, the first-order bias of $\Ah_{\alpha}$, which can be calculated from the parametric bootstrap samples.
Alternatively, we may use a parametric bootstrap to compute $g_{5i}(\Ah_{\alpha})$.
As shown in the proof of Theorem \ref{thm:mse}, $E\{(\thh_i^{\rm R}-\tht_i^{\rm R})(\tht_i^{\rm R}-\tht_i)\}=m^{-1}g_{5i}(A)+o(m^{-1})$, so we can use 
$$
B^{-1}\sum_{b=1}^B\Big\{\thh_i^{\rm R}(y_i^{(b)},\phih_{\alpha}^{(b)})-\tht_i^{\rm R}(y_i^{(b)},\phih_{\alpha})\Big\}\Big\{(\tht_i^{\rm R}(y_i^{(b)},\phih_{\alpha})-\tht_i(y_i^{(b)},\phih_{\alpha}))\Big\}
$$
instead of $m^{-1}g_{5i}(A)$, where $\phih_{\alpha}^{(b)}$ is the parametric bootstrap estimator.
Similarly to Theorem \ref{thm:bMSE}, we can evaluate an error from a finite number of bootstrap iterations, but its evaluation is similar and the detailed proof is omitted.

\section{Examples}\label{sec:num}

\subsection{Simulation studies}\label{sec:comp}
We first investigate the estimation accuracy of the proposed robust estimator together with the existing estimators.
We consider the Fay--Herriot model
$$
y_i=\th_i+\ep_i, \ \ \ \ \ 
\th_i=\beta_0+\beta_1x_i+A^{1/2}u_i, \ \ \ \ i=1,\ldots,m,
$$
where $m=30$, $\beta_0=0$, $\beta_1=2$, $A=$0$\cdot$5, and $\ep_i\sim N(0,D_i)$.
The auxiliary variables $x_i$ are generated from the uniform distribution on $(0,1)$.
Further, we divide $m$ areas into five groups with an equal number of areas and set the same value of $D_i$ within the same groups.
The group $D_i$ pattern is (0$\cdot$2, 0$\cdot$4, 0$\cdot$6, 0$\cdot$8, 1$\cdot$0).
For the distribution of $u_i$, we adopt the structure: $u_i\sim (1-\xi) N(0,1)+\xi N(0,10^2)$, where $c$ determines the degree of misspecification of the assumed distribution (contamination by outliers).  
We consider three scenarios: (I) $\xi=0$, (II) $\xi=$0$\cdot$15, and (III) $\xi=$0$\cdot$30.
Note that, in scenarios (II) and (III), some observations have very large residuals and auxiliary information $x_i$ would be useful for such outlying observations.

We estimate $\th_i$ by using the proposed robust empirical Bayes estimator with density power divergence.
We used two inflation rates, $c=1$ and $c=5$, and $\alpha$ was selected following the procedure given in \S \ref{sec:selection}. 
We adopt four alternative methods: the classical empirical Bayes estimator, the robust Bayes estimator defined in (\ref{SR}) with the model parameters estimated by the robust estimation equation proposed by \cite{Sinha2009} and the maximum likelihood method, and the robust empirical Bayes estimator (\ref{GEB}) proposed by \cite{Ghosh2008} with the maximum likelihood estimator for the model parameters.
Following \cite{Sinha2009}, we set $K$=1$\cdot$345 in Huber's $\psi$-function in equation (\ref{SR}).
A suitable value of $K=K_i$ in (\ref{GEB}) is selected in the same way as in \cite{Ghosh2008} with a $5\%$ inflation rate.

We compute the mean squared errors of those estimators based on 20000 replicates.
Table \ref{tab:comp} reports the values of the mean squared errors averaged within the same groups as well as estimated Monte Carlo errors in the parenthesis.
From the reported values, the Monte Carlo errors seems negligibly small compared with the mean squared errors.
Since the normality assumption in the standard empirical Bayes method is correct in scenario (I), it would be natural that the empirical Bayes method provides smaller mean squared errors than the other methods, but the performance of some robust methods including the proposed method seem comparable with that of the standard method.
On the other hand, there are outlying observations in scenarios (II) and (III), under which the proposed methods tend to produce smaller mean squared errors than the other methods, especially for groups with large sampling variances.
In particular, the performance of the proposed method with $c=5$ is better than that with $c=1$ in these scenarios since the proposed method gets more robust with larger $c$.
However, the performance of the proposed method with $c=5$ is worse than $c=1$ in scenario (I), which would be a reasonable price for the stronger robustness as confirmed in scenarios (II) and (III).
In the Supplementary Material, we provide additional results for other scenarios of $u_i$ such as heavy tailed or skewed distributions.

\vspace{0.5cm}
\begin{table}[!htbp]
\caption{Mean squared errors averaged within the same group. The estimated Monte Carlo errors are reported in the parenthesis. All the values are multiplied by $1000$. DEB1: density power divergence with 1\% inflation rate, DEB2: density power divergence with 5\% inflation rate, EB: empirical Bayes, REB1: robust empirical Bayes of \cite{Sinha2009}, REB2: robust Bayes of \cite{Sinha2009} with maximum likelihood, GEB: robust empirical Bayes of \cite{Ghosh2008}.}
\label{tab:comp}
\begin{center}
\begin{tabular}{cccccccccccccc}
 \hline
 Scenario & Group & DEB1 & DEB2 & EB & REB1 & REB2 & GEB \\
 \hline
 & 1 & 159{\scriptsize (0$\cdot$3)} & 159{\scriptsize (0$\cdot$3)} & 156{\scriptsize (0$\cdot$3)} & 173{\scriptsize (0$\cdot$3)} & 158{\scriptsize (0$\cdot$3)} & 158{\scriptsize (0$\cdot$3)} \\
 & 2 & 256{\scriptsize (0$\cdot$4)} & 258{\scriptsize (0$\cdot$4)} & 252{\scriptsize (0$\cdot$4)} & 280{\scriptsize (0$\cdot$5)} & 258{\scriptsize (0$\cdot$4)} & 309{\scriptsize (0$\cdot$5)} \\
(I) & 3 & 320{\scriptsize (0$\cdot$5)} & 326{\scriptsize (0$\cdot$6)} & 316{\scriptsize (0$\cdot$5)} & 343{\scriptsize (0$\cdot$6)} & 323{\scriptsize (0$\cdot$6)} & 369{\scriptsize (0$\cdot$7)} \\
 & 4 & 356{\scriptsize (0$\cdot$6)} & 366{\scriptsize (0$\cdot$6)} & 352{\scriptsize (0$\cdot$6)} & 378{\scriptsize (0$\cdot$7)} & 359{\scriptsize (0$\cdot$6)} & 393{\scriptsize (0$\cdot$7)} \\
 & 5 & 383{\scriptsize (0$\cdot$7)} & 397{\scriptsize (0$\cdot$7)} & 378{\scriptsize (0$\cdot$7)} & 400{\scriptsize (0$\cdot$7)} & 382{\scriptsize (0$\cdot$6)} & 412{\scriptsize (0$\cdot$7)} \\
 \hline
 & 1 & 186{\scriptsize (0$\cdot$3)} & 179{\scriptsize (0$\cdot$3)} & 192{\scriptsize (0$\cdot$3)} & 328{\scriptsize (5$\cdot$5)} & 189{\scriptsize (0$\cdot$3)} & 189{\scriptsize (0$\cdot$3)} \\
 & 2 & 353{\scriptsize (0$\cdot$6)} & 327{\scriptsize (0$\cdot$6)} & 372{\scriptsize (0$\cdot$6)} & 1017{\scriptsize (12$\cdot$8)} & 362{\scriptsize (0$\cdot$6)} & 364{\scriptsize (0$\cdot$6)} \\
(II) & 3 & 506{\scriptsize (0$\cdot$9)} & 458{\scriptsize (0$\cdot$8)} & 545{\scriptsize (1$\cdot$0)} & 1833{\scriptsize (16$\cdot$9)} & 523{\scriptsize (0$\cdot$9)} & 529{\scriptsize (0$\cdot$9)} \\
 & 4 & 640{\scriptsize (1$\cdot$2)} & 571{\scriptsize (1$\cdot$1)} & 701{\scriptsize (1$\cdot$3)} & 2702{\scriptsize (21$\cdot$3)} & 667{\scriptsize (1$\cdot$2)} & 679{\scriptsize (1$\cdot$2)} \\
 & 5 & 771{\scriptsize (1$\cdot$5)} & 678{\scriptsize (1$\cdot$3)} & 858{\scriptsize (1$\cdot$6)} & 3594{\scriptsize (24$\cdot$7)} & 810{\scriptsize (1$\cdot$5)} & 825{\scriptsize (1$\cdot$5)} \\
 \hline
 & 1 & 194{\scriptsize (0$\cdot$3)} & 190{\scriptsize (0$\cdot$3)} & 196{\scriptsize (0$\cdot$3)} & 223{\scriptsize (2$\cdot$9)} & 195{\scriptsize (0$\cdot$3)} & 195{\scriptsize (0$\cdot$3)} \\
 & 2 & 382{\scriptsize (0$\cdot$6)} & 367{\scriptsize (0$\cdot$6)} & 389{\scriptsize (0$\cdot$7)} & 517{\scriptsize (6$\cdot$2)} & 385{\scriptsize (0$\cdot$6)} & 385{\scriptsize (0$\cdot$6)} \\
(III) & 3 & 562{\scriptsize (1$\cdot$0)} & 534{\scriptsize (0$\cdot$9)} & 578{\scriptsize (1$\cdot$0)} & 949{\scriptsize (9$\cdot$9)} & 568{\scriptsize (1$\cdot$0)} & 568{\scriptsize (1$\cdot$0)} \\
 & 4 & 739{\scriptsize (1$\cdot$2)} & 696{\scriptsize (1$\cdot$2)} & 764{\scriptsize (1$\cdot$3)} & 1518{\scriptsize (14$\cdot$4)} & 748{\scriptsize (1$\cdot$3)} & 747{\scriptsize (1$\cdot$2)} \\
 & 5 & 900{\scriptsize (1$\cdot$5)} & 840{\scriptsize (1$\cdot$5)} & 937{\scriptsize (1$\cdot$6)} & 2253{\scriptsize (18$\cdot$5)} & 911{\scriptsize (1$\cdot$5)} & 913{\scriptsize (1$\cdot$5)} \\
 \hline
\end{tabular}
\end{center}
\end{table}

We next investigate the finite sample performance of the bootstrap estimator of the mean squared error $\hat{M_i}$.
We adopt the same data-generating model with the three scenarios of the distribution of $u_i$ in the previous study with $m=20$.
We also consider the naive estimators of the mean squared error, $\hat{M}_i^{(n1)}$ and $\hat{M}_i^{(n2)}$, obtained by replacing $A$ with $\Ah_{\alpha}$ in the mean squared error formula given in  
Theorems \ref{thm:BEMSE} and \ref{thm:mse}, respectively.
Note that $\hat{M}_i^{(n1)}$ ignores the variability of the estimation of model parameters, and $\hat{M}_i^{(n2)}$ ignores the bias of $g_{1i}(\Ah_{\alpha})+g_{2i}(\Ah_{\alpha})$.  
The motivation using these estimators together with $\hat{M_i}$ is to clarify the importance of the second order unbiasedness under finite sample settings.

We estimate the true values of the mean squared error of the robust empirical Bayes estimator $M_i$ in advance, based on 5000 simulated data. 
The relative bias and square root of the relative mean squared error of the estimator $\hat{M}_i$ are
\begin{align*}
&{\rm RBias}(\hat{M}_i)=100\times E(\hat{M}_i-M_i)/M_i, \\
&{\rm RRMSE}(\hat{M}_i)=100\times E\{(\hat{M}_i-M_i)^2\}/M_i^2.
\end{align*}
These values are computed as averages based on 2000 simulation runs with the bootstrap sample size 1000; they are also averaged within the same groups.

Table \ref{tab:mse-sim} reports the relative biases and square roots of the relative mean squared errors of $\hat{M}_i$, $\hat{M}_i^{(n1)}$, and $\hat{M}_i^{(n2)}$.    
The bootstrap estimator $\hat{M}_i$ outperforms the other estimators owing to the second-order unbiasedness provided in Theorem \ref{thm:bMSE}.
The crude estimators $\hat{M}_i^{(n1)}$ and $\hat{M}_i^{(n2)}$ seem undesirable in practice since they have serious negative biases.

\begin{table}[!htbp]
\caption{Relative bias and root relative mean squared errors for the estimators of the mean squared error.}
\begin{center}
\begin{tabular}{ccccccccccccccccccc}
\hline
& & \multicolumn{3}{c}{RBias}  &   \multicolumn{3}{c}{ RRMSE }\\
Scenario & Group & $\hat{M}_i^{(n1)}$ & $\hat{M}_i^{(n2)}$ & $\hat{M}_i$ & $\hat{M}_i^{(n1)}$ & $\hat{M}_i^{(n2)}$ & $\hat{M}_i$\\
\hline
 & 1 & $-$30$\cdot$9 & $-$14$\cdot$2 & ~ 1$\cdot$3 & 42$\cdot$2 & 19$\cdot$1 & 18$\cdot$3 \\
 & 2 & $-$34$\cdot$6 & $-$20$\cdot$0 & $-$4$\cdot$5 & 46$\cdot$9 & 31$\cdot$3 & 31$\cdot$0 \\
I & 3 & $-$37$\cdot$3 & $-$23$\cdot$5 & $-$8$\cdot$0 & 49$\cdot$8 & 36$\cdot$9 & 36$\cdot$5 \\
 & 4 & $-$34$\cdot$8 & $-$25$\cdot$9 & $-$9$\cdot$2 & 50$\cdot$1 & 42$\cdot$9 & 43$\cdot$4 \\
 & 5 & $-$36$\cdot$8 & $-$26$\cdot$3 & $-$9$\cdot$5 & 51$\cdot$7 & 43$\cdot$5 & 44$\cdot$4 \\
 \hline
 & 1 & ~$-$8$\cdot$3 & ~$-$4$\cdot$4 & ~4$\cdot$3 & 20$\cdot$6 & 13$\cdot$0 & 11$\cdot$4 \\
 & 2 & $-$10$\cdot$9 & ~$-$6$\cdot$7 & ~4$\cdot$1 & 25$\cdot$7 & 20$\cdot$2 & 18$\cdot$3 \\
II & 3 & $-$15$\cdot$1 & $-$10$\cdot$6 & ~1$\cdot$3 & 29$\cdot$8 & 24$\cdot$8 & 22$\cdot$0 \\
 & 4 & $-$12$\cdot$5 & ~$-$9$\cdot$5 & ~4$\cdot$4 & 31$\cdot$7 & 29$\cdot$0 & 28$\cdot$1 \\
 & 5 & $-$16$\cdot$2 & $-$12$\cdot$4 & ~2$\cdot$1 & 34$\cdot$1 & 30$\cdot$8 & 29$\cdot$4 \\
 \hline
 & 1 & ~$-$2$\cdot$8 & ~$-$1$\cdot$7 & ~3$\cdot$8 & 12$\cdot$3 & ~9$\cdot$1 & ~9$\cdot$0 \\
 & 2 & ~$-$5$\cdot$2 & ~$-$3$\cdot$8 & ~3$\cdot$9 & 16$\cdot$0 & 13$\cdot$6 & 12$\cdot$5 \\
III & 3 & ~$-$7$\cdot$1 & ~$-$5$\cdot$4 & ~4$\cdot$1 & 18$\cdot$5 & 16$\cdot$2 & 14$\cdot$9 \\
 & 4 & ~$-$6$\cdot$4 & ~$-$5$\cdot$3 & ~6$\cdot$0 & 20$\cdot$4 & 19$\cdot$1 & 18$\cdot$7 \\
 & 5 & ~$-$9$\cdot$7 & ~$-$8$\cdot$3 & ~4$\cdot$3 & 22$\cdot$6 & 20$\cdot$9 & 19$\cdot$7 \\
 \hline
\end{tabular}
\end{center}
\label{tab:mse-sim}
\end{table}

\subsection{Fresh milk expenditure data}
We consider an application to fresh milk expenditure data from the U.S. Bureau of Labor
Statistics, which was used in \cite{Arora1997} and \cite{YC2006}.
In the data set, the estimated values of the average expenditure on fresh milk for 1989, $y_i$, are available for 43 areas, with the sampling variances $D_i$.
Following \cite{Arora1997}, we consider the Fay--Herriot model (\ref{FH}) with $x_i^\T \beta=\beta_j\ (j=1,\ldots,4)$ if the $i$th area belongs to the $j$th region.
The four regions are $R_1=\{1,\ldots,7\}$, $R_2=\{8,\ldots,14\}$, $R_3=\{15,\ldots,25\}$, and $R_4=\{26,\ldots,43\}$.

Figure \ref{fig:milk} illustrates the scatterplot of $y_i$ with the maximum likelihood estimates of $\beta_1,\ldots,\beta_4$, suggesting that there are some outliers in regions $R_1$ and $R_2$.
To see this, we compute the standardized residuals 
$$
r_i=(\hat{A}+D_i)^{-1/2}\Big\{y_i-\sum_{j=1}^4\beh_jI(i\in M_j)\Big\}, \ \ \ \ \ i=1,\ldots,m.
$$
When the Fay--Herriot model (\ref{FH}) is correctly specified, the distribution of $r_i$ is close to standard normal.
However, as shown in Table \ref{tab:milk-EB}, the absolute values of $r_i$ are high in some areas.

We estimate the parameters using the robust estimation equation of \cite{Sinha2009} and the density power divergence method proposed as the solution to (\ref{EE}) with $1\%$ and $5\%$ inflation rates.
Table \ref{tab:milk} shows that the estimates of $\beta_3$ and $\beta_4$ are similar for the four methods, whereas those of $\beta_1$, $\beta_2$, and $A$ are not necessarily because of the outlying areas in regions $R_1$ and $R_2$.

To estimate $\th_i$, we adopt the classical empirical Bayes estimator $\thh_i^{\rm EB}$, the proposed robust empirical Bayes estimator $\thh_i^{\rm R}$ with a $5\%$ inflation rate, and the robust empirical Bayes estimator $\thh_i^{\rm SR}$ proposed by \cite{Sinha2009}. 
We use $\hat{M}_i$ given in (\ref{bMSE}) to estimate the mean squared error of $\hat{M}_i$ with $B=1000$. 
We then define the mean squared errors of $\thh_i^{\rm EB}$ and $\thh_i^{\rm SR}$ as $M_i^{\rm EB}$ and $M_i^{\rm SR}$, respectively; these are estimated from the result in \cite{Datta2000} for $M_i^{\rm EB}$ and `saeRobust' package in ``R'' for $M_i^{\rm SR}$. 
Table \ref{tab:milk-EB} shows that the differences between $\thh_i^{\rm EB}$ and $\thh_i^{\rm R}$ are large in areas with large absolute values of $r_i$.
Similar phenomena can be observed for the relationship between $\hat{M}_i^{\rm EB}$ and $\hat{M}_i$.
On the contrary, the values of $\thh_i^{\rm SR}$ and $M_i^{\rm SR}$ are different from the others, which might come about from the lower estimate of $A$ as presented in Table \ref{tab:milk}.

\begin{table}[!htbp]
\caption{
Estimates of the model parameters from the four methods. 
Standard errors are shown in the parenthesis.
The estimates and standard errors of $A$ are multiplied by $100$.
}
\begin{center}
\begin{tabular}{ccccccccccccccccccc}
\hline
&& $\beta_1$ & $\beta_2$ & $\beta_3$ & $\beta_4$ & $A$ \\
\hline
Maximum likelihood &  & 0$\cdot$97 & 1$\cdot$10 & 1$\cdot$19 & 0$\cdot$73 & 1$\cdot$55 \\
 &  & (0$\cdot$07) & (0$\cdot$07)& (0$\cdot$06)& (0$\cdot$04) & (0$\cdot$68) \\
Robust maximum likelihood &  & 1$\cdot$01 & 1$\cdot$18 & 1$\cdot$19 & 0$\cdot$73 & 0$\cdot$80 \\
 &  & (0$\cdot$06) & (0$\cdot$07) & (0$\cdot$05) & (0$\cdot$03) & (0$\cdot$53) \\
Density power divergence (1\% inflation) &  & 0$\cdot$97 & 1$\cdot$12 & 1$\cdot$19 & 0$\cdot$73 & 1$\cdot$50 \\
 &  & (0$\cdot$07) & (0$\cdot$07) & (0$\cdot$06) & (0$\cdot$04) & (0$\cdot$65) \\
Density power divergence (5\% inflation) &  & 0$\cdot$98 & 1$\cdot$15 & 1$\cdot$19 & 0$\cdot$73 & 1$\cdot$35 \\
 &  & (0$\cdot$06) & (0$\cdot$07) & (0$\cdot$06) & (0$\cdot$04) & (0$\cdot$60) \\
 \hline
\end{tabular}
\end{center}
\label{tab:milk}
\end{table}

\begin{table}[!htbp]
\caption{Values of the empirical Bayes estimator, proposed robust empirical Bayes estimator, and robust empirical Bayes estimator of \cite{Sinha2009} with their estimates of the mean squared errors.
The values of $\hat{M}_i^{\rm EB}$, $\hat{M}_i$, and $\hat{M}_i^{\rm SR}$ are multiplied by 100. 
}
\begin{center}
\begin{tabular}{cccccccccccccc}
\hline
area & region & $y_i$ & $r_i$ & $\thh_i^{\rm EB}$ & $\thh_i^{\rm R}$ & $\thh_i^{\rm SR}$ &  $\hat{M}_i^{\rm EB}$ & $\hat{M}_i$ & $\hat{M}_i^{\rm SR}$\\
\hline
1 & 1 & 1$\cdot$10 & ~ 0$\cdot$64 & 1$\cdot$02 & 1$\cdot$02 & 1$\cdot$03 & 1$\cdot$35 & 1$\cdot$35 & 0$\cdot$81 \\
4 & 1 & 0$\cdot$63 & $-$2$\cdot$05 & 0$\cdot$78 & 0$\cdot$76 & 0$\cdot$91 & 0$\cdot$85 & 0$\cdot$85 & 4$\cdot$70 \\
5 & 1 & 0$\cdot$75 & $-$1$\cdot$25 & 0$\cdot$86 & 0$\cdot$87 & 0$\cdot$92 & 0$\cdot$96 & 0$\cdot$96 & 4$\cdot$44 \\
9 & 2 & 1$\cdot$41 & ~ 1$\cdot$48 & 1$\cdot$21 & 1$\cdot$24 & 1$\cdot$23 & 1$\cdot$42 & 1$\cdot$40 & 4$\cdot$19 \\
11 & 2 & 0$\cdot$62 & $-$3$\cdot$01 & 0$\cdot$80 & 0$\cdot$73 & 1$\cdot$07 & 0$\cdot$77 & 0$\cdot$78 & 5$\cdot$06 \\
12 & 2 & 1$\cdot$46 & ~ 1$\cdot$54 & 1$\cdot$20 & 1$\cdot$24 & 1$\cdot$23 & 1$\cdot$63 & 1$\cdot$62 & 5$\cdot$66 \\
20 & 3 & 1$\cdot$29 & ~ 0$\cdot$48 & 1$\cdot$23 & 1$\cdot$22 & 1$\cdot$22 & 1$\cdot$31 & 1$\cdot$32 & 0$\cdot$77 \\
25 & 3 & 1$\cdot$19 & $-$0$\cdot$01 & 1$\cdot$19 & 1$\cdot$19 & 1$\cdot$19 & 0$\cdot$81 & 0$\cdot$84 & 0$\cdot$56 \\
31 & 4 & 0$\cdot$89 & ~ 0$\cdot$63 & 0$\cdot$76 & 0$\cdot$76 & 0$\cdot$75 & 1$\cdot$54 & 1$\cdot$63 & 0$\cdot$78 \\
37 & 4 & 0$\cdot$44 & $-$1$\cdot$84 & 0$\cdot$54 & 0$\cdot$54 & 0$\cdot$61 & 0$\cdot$64 & 0$\cdot$65 & 3$\cdot$59 \\
\hline
\end{tabular}
\end{center}
\label{tab:milk-EB}
\end{table}

\begin{figure}
\centering
\includegraphics[width=11cm,clip]{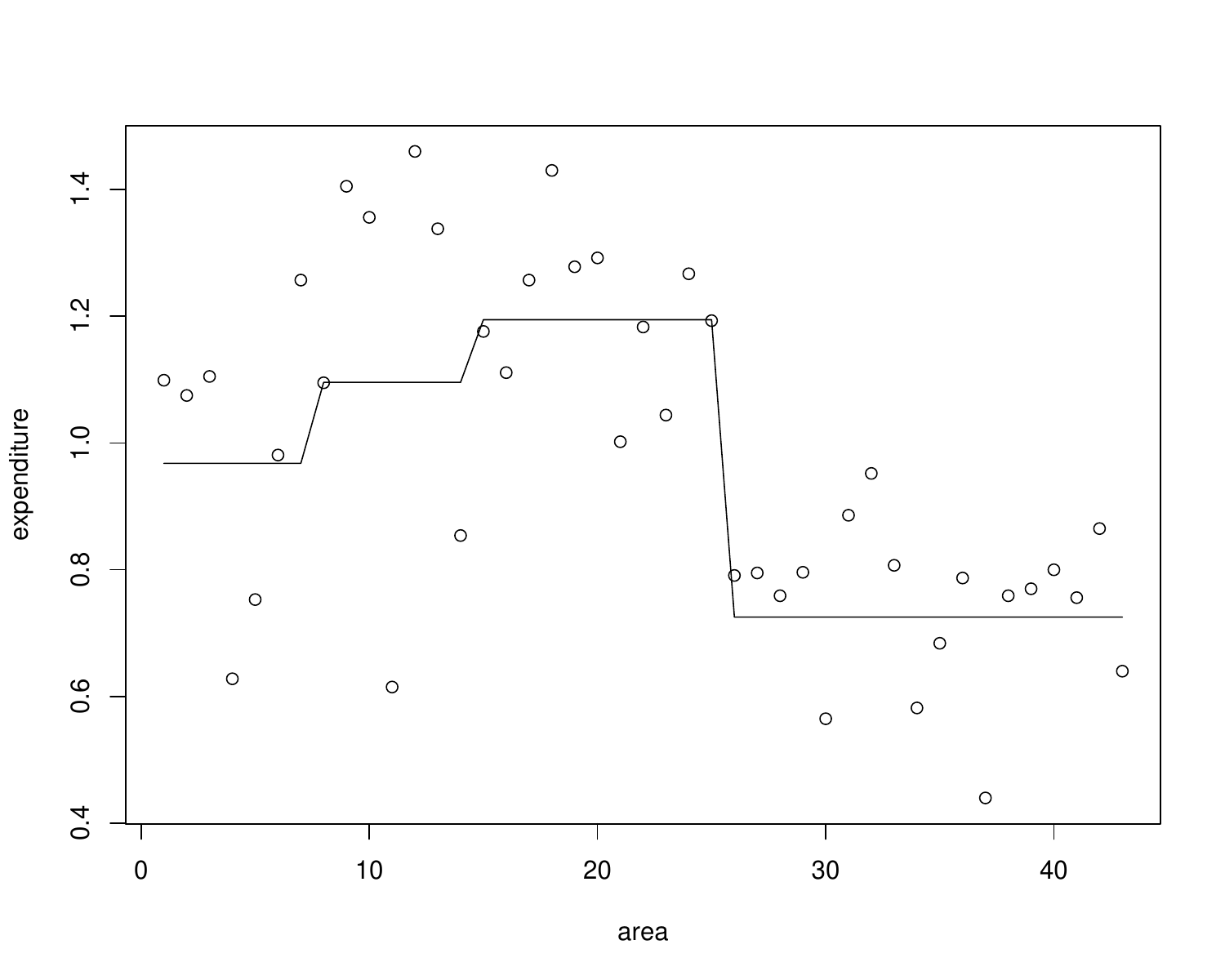} 
\caption{Scatterplot of $y_i$ with the maximum likelihood estimates of $\beta_1,\ldots,\beta_4$. }
\label{fig:milk}
\end{figure}

\section{Final Remarks}\label{sec:disc}
The proposed method would be recommended compared with existing methods especially when there exist outlying observations, as shown in our numerical studies.
Since we revealed some asymptotic properties of the proposed method only under the correct model, investigating asymptotic properties under general model misspecification would be an interesting future work. 
Although this study focused on the Fay--Herriot model, which is standard in small area estimation, the nested error regression model \citep{Battese1998} would be more useful when unit-level data is available.
While several robust methods have been already proposed \citep{Cham2014, Cham2006, Sinha2009}, the extension of the proposed method to unit-level data would be an interesting research direction.
Extending the proposed idea to the non-normal model based on natural exponential family \citep{GM2004} would also be worthwhile.
Finally, several forms of generalized likelihood other than density power divergence have been proposed, such as $\gamma$-divergence \citep{Fujisawa2008}.
The main advantage of density power divergence in this context is its mathematical simplicity.
As presented in \S \ref{sec:DPD}, the robust Bayes predictor has a simple form.
A detailed comparison among generalized likelihood methods is left to a future study.

\section*{Acknowledgments}
The author was supported by the Japan Society of the Promotion of Science (KAKENHI) grant number 18K12757.

\appendix

\section*{Appendix}

{\it Proof of Theorem \ref{thm:BEMSE}.} \ \ \ Since $\tht_i=E(\th_i|y_i)$, 
\begin{align*}
E\{(\tht_i^R-\th_i)^2\}
&=E\{(\tht_i-\th_i)^2\}+E\{(\tht_i^R-\tht_i)^2\}\\
&=\frac{AD_i}{A+D_i}+E\{(\tht_i^R-\tht_i)^2\}\equiv g_{1i}(A)+g_{2i}(A).
\end{align*}
Since $\tht_i^R-\tht_i=(A+D_i)^{-1}D_i(y_i-x_i^\T \beta)(1-s_i)$, 
\begin{equation*}
g_{2i}(A)=\frac{D_i^2}{(A+D_i)^2}E\Big\{(y_i-x_i^\T \beta)^2(1-s_i)^2\Big\}.
\end{equation*}
Using Lemma 1 in the Supplementary Material, we obtain the expression for $g_{2i}(A)$.

We next show that $g_{2i}(A)$ is increasing in $\alpha\in (0,1)$.
For notational simplicity, we put $\mu_i=x_i^\T \beta$.
Since $(y_i-\mu_i)^2(1-s_i)^2$ is a continuous and differentiable function of $y_i$ and $\alpha$, we have
\begin{align*}
\frac{\partial g_{2i}(A)}{\partial\alpha}
&=-\frac{2D_i^2}{(A+D_i)^2}E\Big\{(y_i-\mu_i)^2(1-s_i)\frac{\partial s_i}{\partial \alpha}\Big\}.
\end{align*}
Note that $s_i=f_i(y_i;\phi)^{\alpha}$.
If $f(y_i;\phi)\leq 1$, then $1-s_i\geq 0$ and $s_i$ is decreasing with respect to $\alpha$.
Then, it follows that $(1-s_i)\partial s_i/\partial\alpha\leq 0$. 
On the other hand, we have $(1-s_i)\partial s_i/\partial\alpha\leq 0$ if $f(y_i;\phi)\geq 1$ by a similar argument.
Hence, $(1-s_i)\partial s_i/\partial\alpha\leq 0$ always follows, thereby we have $\partial g_{2i}(A)/\partial \alpha\geq 0$ for $\alpha\in (0,1)$, which completes the proof.

\vspace{0.3cm}
{\it Proof of Theorem \ref{thm:bMSE}.} \ \ \ \ 
It follows that 
\begin{align*}
E\{g_{12i}(\Ah_{\alpha})-g_{12i}(A)\}
&=b_A\frac{\partial g_{12i}(A)}{\partial A}+\frac1{2m}\frac{\partial^2 g_{12i}(A)}{\partial A^2}\frac{K_A}{J_A^2}+\frac16\frac{\partial^3 g_{12i}(A)}{\partial A_{\ast}^3}E\{(\Ah_{\alpha}-A)^3\},
\end{align*}
where $A^{\ast}$ is between $A$ and $\Ah_{\alpha}$.
From Lemma 2 in the Supplementary Material, it holds that $E\{g_{12i}(\Ah_{\alpha})-g_{12i}(A)\}=m^{-1}d(A)+o(m^{-1})$, where $d(\cdot)$ is a smooth function.
Then, from \cite{Butar2003}, we have $E(\Mh_i^{\dagger}-M_i)=o(m^{-1})$.

From the definition of $U_i$, we have
$$
U_i=\frac1B\sum_{b=1}^Bg_{12i}(\Ah_{\alpha}^{(b)})-E^{\ast}\{g_{12i}(\Ah_{\alpha}^{\ast})\},
$$
where $E^{\ast}$ denotes the expectation with respect to the bootstrap sample.
Noting $E(U_i|y)=0$, we observe that 
\begin{align*}
\Var(U_i)&=E\{\Var(U_i|y)\}=\frac1BE[\Var\{g_{12i}(\Ah_{\alpha}^{(1)})|y\}]\\
&=\frac1BE\{E([g_{12i}(\Ah_{\alpha}^{(1)})-E^{\ast}\{g_{12i}(\Ah_{\alpha}^{\ast})\}]^2|y)\}\\
&=\frac1BE\{g_{12i}'(A^{\dagger})^2(\Ah_{\alpha}^{(1)}-A)^2\},
\end{align*}
where $g_{12i}'(A)=\partial g_{12i}(A)/\partial A$ and $A^{\dagger}=\ep A+(1-\ep)\Ah_{\alpha}^{(1)}$ for some $\ep\in [0,1]$.
Straightforward calculation shows that 
\begin{align*}
g_{12i}'(A)=\frac{D_i^2}{(A+D_i)^2}-\frac{g_{2i}(A)}{A+D_i}+\frac{2\pi\alpha D_i}{A+D_i}\left\{\frac{U_i^{\alpha+2}}{(\alpha+1)^{3/2}}-\frac{U_i^{2\alpha+2}}{(2\alpha+1)^{3/2}}\right\}.
\end{align*}
Note that $0\leq U_i\leq (2\pi D_i)^{-1/2}$, thereby $\sup_A|g_{12i}'(A)|\leq C(D_i,\alpha)<\infty$ under Condition 1.
Hence, it follows that 
$$
\Var(U_i)\leq \frac1B C(D_i,\alpha)^2E\{(\Ah_{\alpha}^{(1)}-A)^2\}=O\{(mB)^{-1}\},
$$
which completes the proof.


\setcounter{equation}{0}
\renewcommand{\theequation}{S\arabic{equation}}
\setcounter{section}{0}
\renewcommand{\thesection}{S\arabic{section}}
\setcounter{lem}{0}
\renewcommand{\thelem}{S\arabic{lem}}
\setcounter{table}{0}
\renewcommand{\thetable}{S\arabic{table}}

\newpage
\begin{center}
{\Large{\bf 
Supplementary material for ``Robust Empirical Bayes Small Area Estimation with Density Power Divergence''
}}
\end{center}

\section{Useful Lemma}
In what follows, we use $s_i$ instead of $s_i(y_i;\phi)$ when there is no confusion.
\begin{lem}\label{lem:moment}
When $y_i\sim N(x_i^\T \beta,A+D_i)$, it holds that 
\begin{align*}
&E\{(y_i-x_i^\T \beta)^{2j-1}s_i^k\}=0, \ \ \ \ j,k=1,2,\ldots \\
&E\{(y_i-x_i^\T \beta)^{2j}s_i^k\}=V_i^{k\alpha}(k\alpha+1)^{-j-1/2}(2j-1)!!(A+D_i)^j, \ \ \ j,k=0,1,2,\ldots,
\end{align*}
where $(2j-1)!!=(2j-1)(2j-3)\cdots(1)$. 
\end{lem}

\begin{proof}
Note that
\begin{align*}
E\{(y_i-x_i^\T \beta)^{c}s_i^k\}
&=\frac{V_i^{k\alpha}}{\{2\pi(A+D_i)\}^{1/2}}
\int_{-\infty}^{\infty}(t-x_i^\T \beta)^c\exp\left\{-\frac{(k\alpha+1)(t-x_i^\T \beta)^2}{2(A+D_i)}\right\}\text{d}t\\
&=\frac{V_i^{k\alpha}}{(k\alpha+1)^{1/2}}E(Z^c),
\end{align*}
where $Z\sim N\{0,(A+D_i)/(k\alpha+1)\}$.
Hence, the expectation is $0$ when $c$ is odd.
On the other hand, when $c=2j, \ j=0,1,2,\ldots$, it follows that $E(Z^{2j})=(2j-1)!!(A+D_i)^j(k\alpha+1)^{-j}$, which completes the proof.
\end{proof}

\section{Proof of Theorem 2}
Let $F_\beta$ and $F_A$ be the first and second estimating functions in (10), and defined $F_{\phi}=(F_\beta^\T ,F_A)^\T $
Under Conditions 1-3 in the main article, the theory of unbiased estimating equation by \cite{Godambe1960} shows that $\phih_{\alpha}=(\beh_{\alpha}^\T ,\Ah_{\alpha})^\T $ is consistent and asymptotically normal, with the asymptotic covariance matrix given by
$$
\lim_{m\to\infty}
E\left(\frac1m\frac{\partial F_\phi}{\partial\phi^\T }\right)^{-1}
E\left(\frac1mF_\phi F_\phi^\T\right)
E\left(\frac1m\frac{\partial F_\phi}{\partial\phi^\T }\right)^{-1}.
$$
Straightforward calculation shows that 
\begin{align*}
\frac{\partial F_\beta}{\partial\beta^\T }&
=\sum_{i=1}^m\frac{x_ix_i^\T s_i}{(A+D_i)^2}\Big\{\alpha(y_i-x_i^\T \beta)^2-(A+D_i)\Big\}\\
\frac{\partial F_A}{\partial\beta}&
=\sum_{i=1}^m\frac{x_is_i(y_i-x_i^\T \beta)}{(A+D_i)^3}\Big\{\alpha(y_i-x_i^\T \beta)^2-(\alpha+2)(A+D_i)\Big\}\\
\frac{\partial F_A}{\partial A}&
=\frac12\sum_{i=1}^m\bigg\{\frac{\alpha s_i(y_i-x_i^\T \beta)^4}{(A+D_i)^4}
-\frac{(\alpha^2+2\alpha) V_i^{\alpha}}{(\alpha+1)^{3/2}(A+D_i)^2}
\\
&\hspace{2cm} -\frac{2s_i(\alpha+2)}{(A+D_i)^3}(y_i-x_i^\T \beta)^2
+\frac{s_i(\alpha+2)}{(A+D_i)^2}\bigg\}.
\end{align*}
Then, using lemma \ref{lem:moment}, $E\big(\partial F_\beta/\partial\beta^\T \big)=mJ_{\beta}$, $E\big(\partial F_A/\partial\beta\big)=0$ and $E\big(\partial F_A/\partial A\big)=-mJ_A$.
Moreover, from (12) in the main article and lemma \ref{lem:moment}, $E(F_\beta F_\beta^t)=K_{\beta}$, $E(F_\beta F_A)=0$ and $E(F_A^2)=K_{A}$.
Hence, $\beh_{\alpha}$ and $\Ah_{\alpha}$ is asymptotically independent and their asymptotic covariance matrices are $J_{\beta}^{-1}K_{\beta}J_{\beta}^{-1}$ and $J_A^{-1}K_AJ_A^{-1}$, respectively.

\section{Proof of Theorem 3}
The mean squared error $M_i=E\{(\thh_i^R-\th_i)^2\}$ can be decomposed as 
$$
M_i=E\{(\tht_i^R-\th_i)^2\}+2E\{(\tht_i^R-\th_i)(\thh_i^R-\tht_i^R)\}+E\{(\thh_i^R-\tht_i^R)^2\},
$$
and the first term reduces to $g_{1i}(A)+g_{2i}(A)$ whose expressions are given in Theorem 1.

We first evaluate the third term.
Taylor series expansion shows that 
$$
\thh_i^R-\tht_i^R=\frac{\partial\tht_i^R}{\partial\phi^\T }(\phih_{\alpha}-\phi)+\frac12(\phih_{\alpha}-\phi)^\T \frac{\partial^2\tht_i^R}{\partial\phi_{\ast}\partial\phi_{\ast}^\T }(\phih_{\alpha}-\phi),
$$
where $\phi_{\ast}$ is on the line connecting $\phi$ and $\phih_{\alpha}$. 
Then, we get 
$$
E\{(\thh_i^R-\tht_i^R)^2\}=E\Big[\Big\{\frac{\partial\tht_i^R}{\partial\phi^\T }(\phih_{\alpha}-\phi)\Big\}^2\Big]+R_1+R_2,
$$
where $R_1=E\{(\phih_{\alpha}-\phi)^\T (\partial \tht_i^R/\partial\phi)(\phih_{\alpha}-\phi)^\T (\partial^2 \tht_i^R/\partial\phi_{\ast}\partial\phi_{\ast}^\T )(\phih_{\alpha}-\phi)^\T \}$ and $R_2=E[\big\{(\phih_{\alpha}-\phi)^\T (\partial^2 \tht_i^R/\partial\phi_{\ast}\partial\phi_{\ast}^\T )(\phih_{\alpha}-\phi)^\T \big\}^2]/4$.
Here we use the following lemma.

\begin{lem}\label{lem:order}
Under Conditions 1-3 in the main article, $E(|\phih_{\alpha(k)}-\phi_k|^r)=O(m^{-r/2})$ for any $r>0$ and $k=1,\ldots,p+1$, where $\phih_{\alpha(k)}$ is the $k$th element of $\phih_{\alpha}$.
\end{lem}

A rigorous proof of the lemma requires a uniform integrability, but intuitively, from Theorem 2, $E(m^r|\phih_{\alpha(k)}-\phi_k|^r)=O(1)$ under Conditions 1-3, which leads to Lemma \ref{lem:order}.

In what follows, we use $u_i=y_i-x_i^\T \beta$ and $B_i=A+D_i$ for notational simplicity.
The straightforward calculation shows that 
\begin{align*}
\frac{\partial \tht_i^R}{\partial \beta}
=-\frac{D_is_ix_i}{B_i^2}\left(\alpha u_i^2-B_i\right), \ \ \ \ 
\frac{\partial \tht_i^R}{\partial A}
=-\frac{D_is_iu_i}{2B_i^3}\left\{\alpha u_i^2-(2-\alpha)B_i\right\}.
\end{align*}
Moreover, we have
\begin{align*}
\frac{\partial^2 \tht_i^R}{\partial\beta\partial\beta^\T }
&=-\frac{D_is_ix_ix_i^\T }{B_i^3}(\alpha u_i^3-3B_iu_i)\\
\frac{\partial^2 \tht_i^R}{\partial A^2}
&=\frac{D_is_iu_i}{12B_i^5}\left\{-3\alpha^2 u_i^3+3\alpha^2B_i u_i^2+(4\alpha-3\alpha^2)B_iu_i+(\alpha-2)(3\alpha+8)B_i^2\right\}.
\end{align*}
Note that 
\begin{align*}
R_1&=\sum_{j=1}^{p+1}\sum_{k=1}^{p+1}\sum_{\ell=1}^{p+1}
E\Big\{\Big(\frac{\partial\tht^R_i}{\partial\phi_j}\Big)\Big(\frac{\partial^2\tht^R_i}{\partial\phi_k\partial\phi_{\ell}}\Big)(\phih_{\alpha(j)}-\phi_j)(\phih_{\alpha(k)}-\phi_k)(\phih_{\alpha(\ell)}-\phi_{\ell})\Big\}\\
&\equiv \sum_{j=1}^{p+1}\sum_{k=1}^{p+1}\sum_{\ell=1}^{p+1} U_{1jk\ell}.
\end{align*}
From H\"{o}lder's inequality, 
\begin{align*}
&|U_{1jkl}|
\leq 
E\Big\{\Big|\Big(\frac{\partial\tht^R_i}{\partial\phi_j}\Big)\Big(\frac{\partial^2\tht^R_i}{\partial\phi_k^{\ast}\partial\phi_{\ell}^{\ast}}\Big)\Big|^4\Big\}^{1/4} 
E\Big\{\Big|(\phih_{\alpha(j)}-\phi_j)(\phih_{\alpha(k)}-\phi_k)(\phih_{\alpha(\ell)}-\phi_{\ell})\Big|^{4/3}\Big\}^{3/4}\\
&\ \ \ \ \ \  \leq 
E\Big(\Big|\frac{\partial\tht^R_i}{\partial\phi_j}\Big|^8\Big)^{1/8}
E\Big(\Big|\frac{\partial^2\tht^R_i}{\partial\phi_k^{\ast}\partial\phi_{\ell}^{\ast}}\Big|^8\Big)^{1/8} 
\prod_{a\in \{j,k,\ell\}}E\Big(\Big|\phih_{\alpha(a)}-\phi_a\Big|^{4}\Big)^{1/4}.
\end{align*}
Since $E\big(|\partial\tht^R_i/\partial\phi_j|^8\big)<\infty$ and $E\big(|\partial^2\tht^R_i/\partial\phi_k^{\ast}\partial\phi_{\ell}^{\ast}|^8\big)<\infty$, $R_1=o(m^{-1})$ from Lemma \ref{lem:order}.
A similar evaluation shows that $R_2=o(m^{-1})$.
Using the similar argument given in the proof of Theorem 3 in \cite{Kubokawa2016}, 
\begin{align*}
&E\Big[\Big\{\frac{\partial\tht_i^R}{\partial\phi^\T }(\phih_{\alpha}-\phi)\Big\}^2\Big]
=\tr\Big[E\Big(\frac{\partial\tht_i^R}{\partial\phi}\frac{\partial\tht_i^R}{\partial\phi^\T }\Big)
E\big\{(\phih_{\alpha}-\phi)(\phih_{\alpha}-\phi)^\T \big\}\Big]+o(m^{-1})\\
&\ \ \ \ \ 
=\frac1m\tr\Big\{E\Big(\frac{\partial \tht_i^R}{\partial \beta}\frac{\partial \tht_i^R}{\partial \beta^\T }\Big)J_{\beta}^{-1}K_{\beta}J_{\beta}^{-1}\Big\}
+\frac1mE\Big\{\Big(\frac{\partial \tht_i^R}{\partial A}\Big)^2\Big\}J_A^{-1}K_AJ_A^{-1}+o(m^{-1}).
\end{align*}
From Theorem 2 and 
\begin{align*}
&E\Big(\frac{\partial \tht_i^R}{\partial \beta}\frac{\partial \tht_i^R}{\partial \beta^\T }\Big)
=\frac{D_i^2V_i^{2\alpha}x_ix_i^\T }{(A+D_i^2)(2\alpha+1)^{3/2}}\\
&E\Big\{\Big(\frac{\partial \tht_i^R}{\partial A}\Big)^2\Big\}
=\frac{D_i^2V_i^{2\alpha}}{(A+D_i)^3(2\alpha+1)^{7/2}}\big(\alpha^4-\frac12\alpha^2+1\big),
\end{align*}
we obtain $E\{(\thh_i^R-\tht_i^R)^2\}=m^{-1}g_{3i}(A)+m^{-1}g_{4i}(A)+o(m^{-1})$.

Concerning $E\{(\tht_i^R-\th_i)(\thh_i^R-\tht_i^R)\}$, 
\begin{align*}
E\{(\tht_i^R-\th_i)(\thh_i^R-\tht_i^R)\}
&=E\{(\tht_i^R-\tht_i)(\thh_i^R-\tht_i^R)\}=\frac{D_i}{B_i}E\{(1-s_i)u_i(\thh_i^R-\tht_i^R)\}.
\end{align*}
By Taylor series expansion, 
\begin{align*}
\thh_i^R-\tht_i^R
=\frac{\partial\tht_i^R}{\partial\phi^\T }(\phih_{\alpha}-\phi)
+\frac12(\phih_{\alpha}-\phi)^\T \frac{\partial^2\tht_i^R}{\partial\phi\partial\phi^\T }(\phih_{\alpha}-\phi)+R_3,
\end{align*}
where 
$$
R_3=\frac16\sum_{k=1}^{p+1}\sum_{j=1}^{p+1}\sum_{\ell=1}^{p+1}\frac{\partial^3\tht_i^R}{\partial\phi_k^{\ast}\partial\phi_j^{\ast}\partial\phi_\ell^{\ast}}(\phih_{\alpha(k)}-\phi_k)(\phih_{\alpha(j)}-\phi_j)(\phih_{\alpha(\ell)}-\phi_\ell).
$$
Similarly to the evaluation of $R_1$, $E\{(1-s_i)u_iR_3\}=o(m^{-1})$.
Then, 
\begin{align*}
E&\{(\tht_i^R-\th_i)(\thh_i^R-\tht_i^R)\}\\
&=\frac{D_i}{B_i}E\Big\{(1-s_i)u_i\frac{\partial\tht_i^R}{\partial\phi^\T }(\phih_{\alpha}-\phi)\Big\}
+\frac{D_i}{2B_i}E\Big\{(1-s_i)u_i(\phih_{\alpha}-\phi)^\T \frac{\partial^2\tht_i^R}{\partial\phi\partial\phi^\T }(\phih_{\alpha}-\phi)\Big\}+o(m^{-1})\\
&\equiv T_1+T_2+o(m^{-1}),
\end{align*}
where
\begin{equation*}
\begin{split}
T_2&=\frac{D_i}{2mB_i}\tr\Big[E\Big\{(1-s_i)u_i\frac{\partial^2\tht_i^R}{\partial\beta\partial\beta^\T }\Big\}J_{\beta}^{-1}K_{\beta}J_{\beta}^{-1}\Big]\\
& \ \ \ +\frac{D_iK_A}{2mB_iJ_A^2}E\Big\{(1-s_i)u_i\frac{\partial^2\tht_i^R}{\partial A^2}\Big\}+o(m^{-1}).
\end{split}
\end{equation*}
From \cite{Lohr2009}, 
\begin{align*}
E(\beh_{\alpha}-\beta|y_i)&=b_{\beta}-m^{-1}B_i^{-1}J_{\beta}^{-1}x_is_iu_i+o_p(m^{-1})\\ 
E(\Ah_{\alpha}-A|y_i)&=b_A-m^{-1}J_A^{-1}\Big\{\frac{u_i^2s_i}{B_i^2}-\frac{s_i}{B_i}+\frac{\alpha V_i^{\alpha}}{(\alpha+1)^{3/2}B_i}\Big\}+o_p(m^{-1}),
\end{align*}
where $b_{\beta}=\lim_{m\to\infty}mE(\beh_{\alpha}-\beta)$ and $b_A=\lim_{m\to\infty}mE(\Ah_{\alpha}-A)$, so 
\begin{equation*}
\begin{split}
T_1&=-\frac{D_i}{mB_i^2}E\Big\{s_i(1-s_i)u_i^2\frac{\partial\tht_i^R}{\partial\beta^\T }J_{\beta}^{-1}x_i\Big\}
-\frac{D_i}{mB_i^3J_A}E\Big\{s_i(1-s_i)u_i\frac{\partial\tht_i^R}{\partial A}(u_i^2-B_i)\Big\}\\
& \ \ \ \ +\frac{D_i}{B_i}E\Big\{(1-s_i)u_i\frac{\partial\tht_i^R}{\partial A}\Big\}
\Big\{b_A-\frac{\alpha V_i^{\alpha}}{m(\alpha+1)^{3/2}B_iJ_A}\Big\}+o(m^{-1}).
\end{split}
\end{equation*}
Combining these results and using Lemma \ref{lem:moment}, $E\{(\tht_i^R-\th_i)(\thh_i^R-\tht_i^R)\}=m^{-1}g_{5i}(A)+o(m^{-1})$, which completes the proof.

\section{Additional simulation study}
We show results of additional simulation studies regarding estimation accuracy of several estimators of $\th_i$.
We use the same data generating model for $y_i$ as in Section 4$\cdot$1 in the main article. 
We consider the following additional scenarios of the true generating distribution of $u_i$:
\begin{align*}
&\text{(IV)}\ \ u_i\sim t_2, \ \ \ \ 
\text{(V)}\ \  u_i\sim {\rm Ga}(0.5,0.5),\\
&\text{(VI)} \ \ u_i\sim {\rm Ga}(2,2), \ \ \ \ \ \ 
\text{(VII)} \ \ u_i\sim {\rm ST}_3(2)/3^{1/2},
\end{align*}
where $t_n$ is a $t$-distribution with $n$ degrees of freedom,  ${\rm Ga}(a,b)$ is a gamma distribution with shape parameter $a$ and rate paramour $b$, which are scaled to have mean zero and variance $1$, and ${\rm ST}_n(a)$ denotes a skew $t$-distribution with $n$ degrees of freedom  and skewing parameter $a$ in the parametrization given in `skewt' package in ``R".
For estimating $\th_i$, we employ the same six methods used in the main article, and compute mean squared errors based on 20000 replicates.
Table \ref{tab:comp-add} reports the values of mean squared errors averaged within the same groups.
Since we found that the estimated Monte Carlo errors are negligibly small as given in Table 1 in the main article, they are not shown here.
In scenario (IV) and (VII), the generated values of $u_i$ sometimes contain outliers due to the heavy tailed properties of $t$- or skew $t$-distributions, under which the proposed methods tend to provide better performance than the other methods.
Note that the true distributions of $u_i$ in scenarios (V) and (VI) are skewed, but it would not produce extreme values.
In particular, the distribution of $u_i$ in scenario (V) is more skewed than scenario (VI), and the proposed methods tend to perform better than the other methods in scenario (V).
On the other hand, the standard empirical Bayes method performs quite well in scenario (VI) in spite of misspecification of the normality assumption, and the proposed method is comparable or slightly better than the other robust methods.

\begin{table}
\caption{Simulated mean squared errors averaged within the same group. The values are multiplied by $1000$. DEB1: density power divergence with 1\% inflation rate, DEB2: density power divergence with 5\% inflation rate, EB: empirical Bayes, REB1: robust empirical Bayes of \cite{Sinha2009}, REB2: robust Bayes of \cite{Sinha2009} with maximum likelihood, GEB: robust empirical Bayes of \cite{Ghosh2008}.
\label{tab:comp-add}}
\begin{center}
\begin{tabular}{cccccccccccccc}
 \hline
Scenario & Group & DEB1 & DEB2 & EB & REB1 & REB2 & GEB \\
\hline
& 1 & 180 & 178 & 183 & 249 & 180 & 180 \\
 & 2 & 329 & 323 & 340 & 591 & 332 & 337 \\
(IV) & 3 & 456 & 446 & 478 & 1250 & 464 & 492 \\
 & 4 & 563 & 550 & 597 & 1402 & 578 & 621 \\
 & 5 & 661 & 644 & 712 & 2004 & 691 & 732 \\
 \hline
 & 1 & 145 & 143 & 145 & 216 & 143 & 150 \\
 & 2 & 232 & 229 & 234 & 334 & 235 & 266 \\
(V) & 3 & 293 & 289 & 297 & 397 & 304 & 321 \\
 & 4 & 335 & 332 & 340 & 427 & 354 & 352 \\
 & 5 & 365 & 363 & 372 & 449 & 392 & 371 \\
 \hline
 & 1 & 156 & 156 & 154 & 183 & 154 & 157 \\
 & 2 & 249 & 249 & 247 & 296 & 250 & 293 \\
(VI) & 3 & 310 & 312 & 308 & 358 & 317 & 351 \\
 & 4 & 352 & 358 & 351 & 395 & 362 & 380 \\
 & 5 & 378 & 387 & 377 & 415 & 389 & 398 \\
 \hline
 & 1 & 163 & 161 & 164 & 231 & 162 & 163 \\
 & 2 & 279 & 274 & 284 & 509 & 280 & 311 \\
(VII) & 3 & 363 & 355 & 374 & 703 & 371 & 411 \\
 & 4 & 435 & 426 & 452 & 771 & 455 & 476 \\
 & 5 & 489 & 479 & 512 & 895 & 520 & 523 \\
 \hline
\end{tabular}
\end{center}

\vspace{0.2cm}

\end{table}



\begin{thebibliography}{100}
\expandafter\ifx\csname natexlab\endcsname\relax\def\natexlab#1{#1}\fi


\bibitem[{Agostinelli \& Greco(2013)}]{AG2013}
\textsc{Agostinelli, G.} \& \textsc{Greco, L.} (2013).
\newblock{A weighted strategy to handle likelihood uncertainty in Bayesian inference.}
\newblock \textit{Comput. Stat.}  {\bf 28}, 319--339.



\bibitem[{Arora \& Lahiri(1997)}]{Arora1997}
\textsc{Arora, V.} \& \textsc{Lahiri, P.} (1997).
\newblock{On the superiority of the Bayesian method over the BLUP in small area estimation problems.}
\newblock \textit{Statist. Sinica} \textbf{7}, 1053--1063.



\bibitem[{Basu et al.(1998)}]{Basu1998}
\textsc{Basu, A.}, \textsc{Harris, I. R.}, \textsc{Hjort, N. L.} \& \textsc{Jones, M. C.} (1998).
\newblock{Robust and efficient estimation by minimizing a density power divergence.}
\newblock \textit{Biometrika} \textbf{85}, 549--559.





\bibitem[{Batesse et al.(1988)}]{Battese1998}
\textsc{Battese, G.E.}, \textsc{Harter, R.M.} \& \textsc{Fuller, W.A.} (1988). 
\newblock{An error-components model for prediction of county crop areas using survey and satellite data.}
\newblock \textit{J. Am. Statist. Assoc.} \textbf{83}, 28--36.



\bibitem[{Burden \& Faires, 2010}]{Burden2010}
\textsc{Burden, R. L.} \& \textsc{Faires, J. D.} (2010). 
\newblock \textit{Numerical Analysis, 9th Edition}, Stanford: Brooks-Cole Publishing.



\bibitem[{Butar \& Lahiri(2003)}]{Butar2003}
\textsc{Butar, F. B.} \& \textsc{Lahiri, P.} (2003).
\newblock{On measures of uncertainty of empirical Bayes small-area estimators.}
\newblock \textit{J. Stat. Plan. Infer.} \textbf{12}, 63--76.




\bibitem[Carvalho et al.(2010)]{Cal2010}
\textsc{Carvalho, C. M.}, \textsc{Polson, N. G.} \& \textsc{Scott, J. G.} (2010).
\newblock{The horseshoe estimator for sparse signals.}
\newblock \textit{Biometrika} \textbf{97}, 465--480.







\bibitem[{Chambers et al.(2014)}]{Cham2014}
\textsc{Chambers, R. L.}, \textsc{Chandra, H.}, \textsc{Salvati, N.} \& \textsc{Tzavidis, N.} (2014).
\newblock{Outlier robust small area estimation.}
\newblock \textit{J. R. Stat. Soc. B.} {\bf 76}, 47--69.




\bibitem[{Chambers \& Tzavidis(2006)}]{Cham2006}
\textsc{Chambers, R. L.} \& \textsc{Tzavidis, N.} (2006).
\newblock{M-quantile models for small area estimation.}
\newblock \textit{Biometrika} {\bf 93}, 255--268. 




\bibitem[{Chang \& Hall(2015)}]{CH2015}
\textsc{Chang, J.} \& \textsc{Hall, P.} (2015).
\newblock{Double-bootstrap methods that use a single double-bootstrap simulation.}
\newblock \textit{Biometrika} {\bf 102}, 203--214. 






\bibitem[{Datta \& Lahiri(1995)}]{Datta1995}
\textsc{Datta, G. S.} \& \textsc{Lahiri, P.} (1995). 
\newblock{Robust hierarchical Bayes estimation of small area characteristics in the presence of covariates and outliers.}
\newblock {\it J. Multivariate Anal.} {\bf 54}, 310--328.



\bibitem[{Datta \& Lahiri(2000)}]{Datta2000}
\textsc{Datta, G. S.} \& \textsc{Lahiri, P.} (2000). 
\newblock{A unified measure of uncertainty of estimated best linear unbiased predictors in small area estimation problems. }
\newblock {\it Stat. Sinica.} {\bf 10}, 613--627.




\bibitem[{Datta et al.(2005)}]{Datta2005}
\textsc{Datta, G.S.}, \textsc{Rao, J.N.K.} \& \textsc{Smith, D.D.} (2005).
\newblock{On measuring the variability of small area estimators under a basic area level model.}
\newblock {\it Biometrika} {\bf 92}, 183--196.





\bibitem[{Efron, 2011}]{Efron2011}
\textsc{Efron, B.} (2011).
\newblock {Tweedie's formula and selection bias}.
\newblock \textit{J. Am. Stat. Assoc.} {\bf 106}, 1602--1614.







\bibitem[{Fay \& Herriot(1979)}]{FH1979}
\textsc{Fay, R. E.} \& \textsc{Herriot, R. A.} (1979).
\newblock{Estimates of income for small places: an application of James--Stein procedures to census data.}
\newblock \textit{J. Am. Stat. Assoc.} {\bf 74}, 269--277.





\bibitem[{Fujisawa \& Eguchi(2008)}]{Fujisawa2008}
\textsc{Fujisawa, H.} \& \textsc{Eguchi, S.} (2008).
\newblock{Robust parameter estimation with a small bias against heavy contamination}
\newblock {\it J. Multivariate Anal.} {\bf 99}, 2053--2081.





\bibitem[{Ghosh \& Basu(2013)}]{GB2013}
\textsc{Ghosh, A.} \& \textsc{Basu, A.} (2013).
\newblock{Robust estimation for independent non-homogeneous observations using density power divergence with applications to linear regression.}
\newblock \textit{Electr. J. Stat.} {\bf 7}, 2420--2456.



\bibitem[{Ghosh \& Basu(2016)}]{GB2016}
\textsc{Ghosh, A.} \& \textsc{Basu, A.} (2016).
\newblock{Robust Bayes estimation using the density power divergence. }
\newblock \textit{Ann. Inst. Stat. Math.} {\bf 68}, 413--437.



\bibitem[{Ghosh \& Maiti(2004)}]{GM2004}
\textsc{Ghosh, M.} \& \textsc{Maiti, T.} (2004).
\newblock{Small-area estimation based on natural exponential family quadratic variance function models and survey weights. }
\newblock \textit{Biometrika} {\bf 91}, 95--112.





\bibitem[{Ghosh et al.(2008)}]{Ghosh2008}
\textsc{Ghosh, M.}, \textsc{Maiti, T.} \& \textsc{Roy, A.} (2008).
\newblock{Influence functions and robust Bayes and empirical Bayes small area estimation.}
\newblock {\it Biometrika} {\bf 95}, 573--585.




\bibitem[{Hall \& Maiti(2006)}]{Hall2006}
\textsc{Hall, P.} \& \textsc{Maiti, T.} (2006).
\newblock{On parametric bootstrap methods for small area prediction.}
\newblock {\it J. R. Stat. Soc. B.} {\bf 68}, 221--238.




\bibitem[{Hooker \& Vidyashankar(2014)}]{HV2014}
\textsc{Hooker, G.} \& \textsc{Vidyashankar, A. B.} (2014).
\newblock{Bayesian model robustness via disparities.}
\newblock {\it TEST}, {\bf 23}, 556--584.




\bibitem[{Huber's(1973)}]{Huber1973}
\textsc{Huber, P. J.} (1973).
\newblock{Robust regression: asymptotics, conjectures and Monte Carlo.}
\newblock \textit{Ann. Stat.} {\bf 1}, 799--821.




\bibitem[{Jewson et al.(2018)}]{Jewson2018}
\textsc{Jewson, J.}, \textsc{Smith, J. Q.} \& \textsc{Holmes, C.} (2018).
\newblock{Principles of Bayesian inference using general divergence criteria.}
\newblock \textit{Entropy} {\bf 20}, 442.





\bibitem[{Nakagawa \& Hashimoto(2019)}]{NH2019}
\textsc{Nakagawa, T.} \& \textsc{Hashimoto, S.} (2019).
\newblock{Robust Bayesian inference via $\gamma$-divergence.}
\newblock {\it Commun. Stat. Theory.}, to appear.




\bibitem[{Pfeffermann(2013)}]{Pfe2013}
\textsc{Pfeffermann, D.} (2013).
\newblock{New important developments in small area estimation. }
\newblock \textit{Stat. Sci.} {\bf 28}, 40--68.




\bibitem[{Prasad \& Rao(1990))}]{Prasad1990}
\textsc{Prasad, N.} \& \textsc{Rao, J. N. K.} (1990).
\newblock{The estimation of mean-squared errors of small-area estimators.}
\newblock \textit{J. Am. Stat. Assoc.} {\bf 90}, 758--766.





\bibitem[{Rao \& Molina(2015)}]{Rao2015}
\textsc{Rao, J.N.K.} \& \textsc{Molina, I.} (2015).
\newblock \textit{Small Area Estimation, 2nd Edition}, New York: Wiley.






\bibitem[{Sinha \& Rao(2009))}]{Sinha2009}
\textsc{Sinha, S. K.} \& \textsc{Rao, J. N. K.} (2009).
\newblock{Robust small area estimation.}
\newblock \textit{Can. J. Stat.} {\bf 37}, 381--399.




\bibitem[{You \& Chapman(2006))}]{YC2006}
\textsc{You, Y.} \& \textsc{Chapman, B.} (2006).
\newblock{Small area estimation using area level models and estimated sampling
variances.}
\newblock \textit{Surv. Method.} {\bf 32}, 97--103.



\end{thebibliography}

\begin{thebibliography}{100}
\expandafter\ifx\csname natexlab\endcsname\relax\def\natexlab#1{#1}\fi


%



\bibitem[{Ghosh et al.(2008)}]{Ghosh2008}
\textsc{Ghosh, M.}, \textsc{Maiti, T.} \& \textsc{Roy, A.} (2008).
\newblock{Influence functions and robust Bayes and empirical Bayes small area estimation.}
\newblock {\it Biometrika} {\bf 95}, 573--585.




%
\bibitem[{Godambe(1960)}]{Godambe1960}
\textsc{Godambe, V. P.} (1960).
\newblock{An optimum property of regular maximum likelihood estimation.}
\newblock \textit{Ann. Math. Stat.}  \textbf{31}, 1208--1211.



%
\bibitem[{Kubokawa et al.(2016)}]{Kubokawa2016}
\textsc{Kubokawa, T.}, \textsc{Sugasawa, S.}, \textsc{Ghosh, M.} \& \textsc{Chaudhuri, S.} (2016).
\newblock{Prediction in heteroscedastic nested error regression models with random dispersions.}
\newblock \textit{Stat. Sinica}  {\bf 26}, 465--492.



%



%
\bibitem[{Lohr and Rao(2009)}]{Lohr2009}
\textsc{Lohr, S. L.} \& \textsc{Rao, J. N. K.} (2009).
\newblock{Jackknife estimation of mean squared error of small area predictors in nonlinear mixed models.}
\newblock  {\it Biometrika} {\bf 96}, 457--468.




\bibitem[{Sinha \& Rao(2009))}]{Sinha2009}
\textsc{Sinha, S. K.} \& \textsc{Rao, J. N. K.} (2009).
\newblock{Robust small area estimation.}
\newblock \textit{Can. J. Stat.} {\bf 37}, 381--399.

\end{thebibliography}
\end{document}